\documentclass[prx, aps, onecolumn, superscriptaddress, notitlepage]{revtex4-1}
\usepackage{amsthm, color, bbm, graphicx, times, amsfonts, amsmath}
\usepackage[unicode=true, bookmarks=true, bookmarksnumbered=false, bookmarksopen=false, breaklinks=false, pdfborder={0 0 1}, backref=false, colorlinks=false]{hyperref}

{\catcode`\|=\active 
	\gdef\Braket#1{\begingroup
		\mathcode`\|32768\let\vert\BraVert\left<{#1}\right>\endgroup}}
\def\BraVert{\egroup\,\mid\,\bgroup}

\def\tr{\mbox{tr}}

\definecolor{Blue}{rgb}{0,0,1}
\definecolor{Red}{rgb}{0,0,0}
\definecolor{Green}{rgb}{0,1,0}
\definecolor{Purp}{rgb}{.2,0,.2}
\definecolor{white}{rgb}{1,1,1}

\makeatletter
\theoremstyle{plain}
\newtheorem{thm}{\protect\theoremname}
\providecommand{\theoremname}{Theorem}

\theoremstyle{plain}

\providecommand{\lemmaname}{Lemma}

\theoremstyle{plain}
\newtheorem{prop}[thm]{\protect\propositionname}
\providecommand{\propositionname}{Proposition}

\theoremstyle{plain}
\newtheorem{defn}[thm]{\protect\definitionname}
\providecommand{\definitionname}{Definition}

\theoremstyle{plain}

\providecommand{\conjecturename}{Conjecture}

\begin{document}
\title{Divisible quantum dynamics satisfies temporal Tsirelson's bound}

\author{Thao Le}
\affiliation{School of Physics \& Astronomy, Monash University, Victoria 3800, Australia }

\author{Felix A. Pollock}
\affiliation{School of Physics \& Astronomy, Monash University, Victoria 3800, Australia }

\author{Tomasz Paterek}
\affiliation{School of Physical and Mathematical Sciences, Nanyang Technological University, Singapore and Centre for Quantum Technologies, National University of Singapore, Singapore}

\author{Mauro Paternostro}
\affiliation{School of Mathematics and Physics, Queen’s University, Belfast BT7 1NN, United Kingdom}

\author{Kavan Modi}
\affiliation{School of Physics \& Astronomy, Monash University, Victoria 3800, Australia }
\email{kavan.modi@monash.edu}


\begin{abstract}
We give strong evidence that divisibility of qubit quantum processes implies temporal Tsirelson's bound. We also give strong evidence that the classical bound of the temporal Bell's inequality holds for dynamics that can be described by entanglement-breaking channels---a more general class of dynamics than that allowed by classical physics.
\end{abstract}

\maketitle

Two classical systems interrogated by space-like separated measurements give rise to correlations bounded by Bell's inequalities~\cite{Bell1994}. Remarkably, quantum systems can violate such bounds, although they cannot achieve the maximal algebraically allowed value~\cite{PRboxes1994}. The quantum maximum, dubbed Tsirelson's bound~\cite{Tsirelson1980}, stems from reasons that are now well understood: violation of this bound would trivialise communication complexity~\cite{vandam,PRL.96.250401} and be against a number of natural postulates~\cite{infocaus, ProcRolSocA.466.881, NJP.14.063024, FoundPhys.43.805, D.14.239, PRL.112.040401, arXiv1507.07514}. In a different yet related context, a number of works have studied correlations between the outcomes of time-like separated observables~\cite{LG1985, Taylor2004, EurophysLett.75.202, PhysRevA.80.034102, Fritz2010, Emary2004LGI, PhysRevA.82.030102, Emary2013, PhysRevLett.111.020403, PhysRevA.89.062319, PhysRevLett.115.120404}. In this scenario, the reasons behind the existence of a Tsirelson-like bound, limiting the value taken by suitably built functions of two-time correlators, are not as clear. In this paper we shed light on this fundamental question, giving strong evidence that Tsirelson's bound for temporal correlations follows from a well-known and prevalent property of dynamical processes, namely their {\it divisibility} (see Ref.~\cite{arXiv:1308.5761} for related first investigations).

Divisibility asserts that dynamical evolution between any two points in time can be decomposed into a series of intermediate-time evolutions. Fundamental dynamics is expected to be divisible and, indeed, the Sch\"odinger equation generates unitary processes, which are fully divisible. Furthermore, divisible evolutions are often good approximations to open-system dynamics and divisibility is assumed explicitly in the derivation of several master equations~\cite{BreuerPetruccione}. In fact, whenever the Markov assumption holds, i.e., the evolving system is memory-less, the process is divisible~\cite{arXiv:1512.00589} and conversely divisible channels always decrease information~\cite{PhysRevA.93.012101}. This provides an intuition as to why divisibility might be the relevant feature for temporal Bell's inequalities. It is known that, in the temporal setting, both classical and quantum bounds on the temporal Bell's inequalities can be violated even with purely classical systems if they embody sufficiently large memory~\cite{PhysRevLett.115.120404, NJP.13.113011}, thus effectively breaking the divisibility condition. An explicit example of this will be given later on in this paper. We also give strong evidence that the usual ``classical'' bound on the temporal Bell's inequality holds for a more general class of dynamics than stochastic maps consistent with classical physics. This parallels the situation for space-like measurements, where the classical bound on Bell's inequality holds for local hidden variable models. These are strictly richer than classical ones, as illustrated for example in Ref.~\cite{PRL.93.230403}, where imposing invariance of measured correlations under rotations of local coordinate systems is shown to lead to a more stringent version of Bell's theorem.

\begin{figure}[!b]
\centering \includegraphics[width=.7\columnwidth]{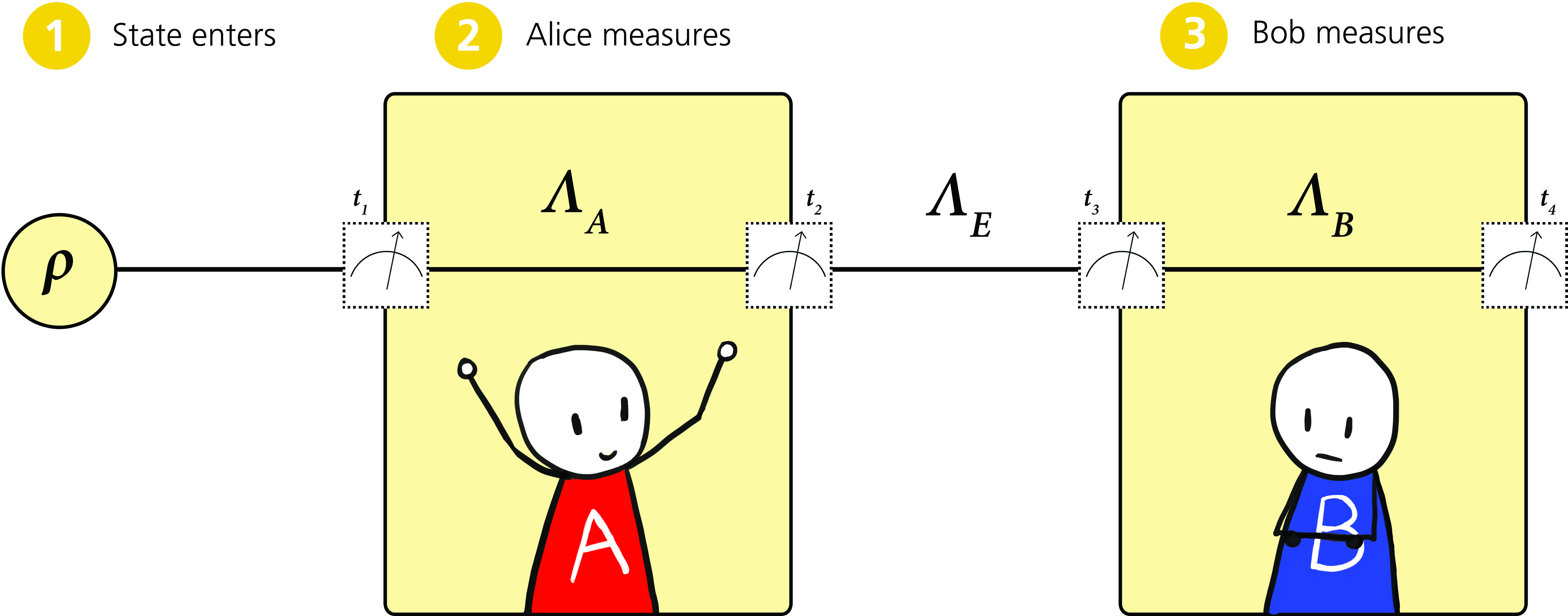} \protect\caption{Generalised temporal Clauser-Horne-Shimony-Holt scenario. The $\Lambda_{i}$ are arbitrary channels (completely positive trace preserving maps). A physical system in a state $\rho$ is first measured by Alice either at time $t_1$ or $t_2$, and then by Bob either at time $t_3$ or $t_4$. We give strong evidence that the correlations observed by Alice and Bob, i.e., expectation values of the product of their results, satisfy temporal Tsirelson's bound whenever the $\Lambda$s are independent and hence the dynamics is divisible (see Def.~\ref{defn:divisible}).} \label{fig:scenario}
\end{figure}

\section{Scenario}

Consider the situation depicted in Fig.~\ref{fig:scenario}, where two observers, Alice and Bob, make time-ordered measurements with Alice measuring before Bob. Each choose to measure at one of two times; Alice (Bob) measures either at time $t_1$ or $t_2$ ($t_3$ or $t_4$). We allow for intermediate dynamics between any consecutive measurement times, and label the corresponding general quantum channels as $\Lambda_A$ ($\Lambda_B$) for the evolution between $t_1$ and $t_2$ ($t_3$ and $t_4$), and $\Lambda_E$ for the dynamics between $t_2$ and $t_3$ (c.f. Fig.~\ref{fig:scenario}). From their measurement outcomes the following temporal Bell function is constructed
\begin{gather}
\label{bellfunction}
\mathcal{B} = E_{13} + E_{14} + E_{23} - E_{24}.
\end{gather}
Throughout this paper we will be calculating the temporal Bell function above with various assumptions and restrictions placed on $\Lambda_A$, $\Lambda_E$, and $\Lambda_B$. Note that channels $\Lambda_A$, $\Lambda_E$, and $\Lambda_B$ individually may not be divisible. That is, we only care about divisibility between the labs of Alice and Bob. Such a process is called 3-divisible~\cite{RPP.77.094001}.

Above, $E_{ij}$ is the correlation function between the $i$th measurement performed by Alice and the $j$th one by Bob, with $i,j$ denoting the instant of time at which measurements are performed, i.e., $i \in \{1,2\}$ and $j \in \{3,4\}$, which we call {\it time steps}. The correlation functions are defined as the expectation value of the product of measurement results obtained by Alice and Bob. They are calculated under the assumption that every experimental run is an independent event, i.e., without allowing for adaptive strategies where the measurement choices in a given run would depend on the outcomes obtained in previous experimental runs \footnote{In the temporal setting of the CHSH game, Alice and Bob can violate Tsirelson’s bound. However, this requires Bob to choose his inputs and outputs based on the previous run, which leads to indivisibility of the process.}. We consider dichotomic $\pm 1$ observables parametrised by their corresponding Bloch vectors $\vec \sigma \cdot \vec a_i$ and $\vec \sigma \cdot \vec b_j$ of Alice and Bob respectively. The initial state is parametrised by $\rho = \frac{1}{2} (\mathbbm{1} + \vec \sigma \cdot \vec v)$.

Note that our model generalises that of Ref.~\cite{Taylor2004} and reduces to their model when there is no dynamics between the two measurement choices of Alice and Bob. In this case, the temporal correlations can be turned into spatial correlations using the Choi-Jamio{\l}kowski isomorphism; one can therefore resort to standard tools to recover Tsirelson's bound. However, in our scenario the correlation functions for different settings are measured on different states---a consequence of the different channels that act between each pair of measurement times. Our scenario is therefore richer, despite the fact that Alice and Bob perform the same number of measurements. For instance, Tsirelson's bound cannot be violated in the model of Ref.~\cite{Taylor2004}, while it can be in our model (see Prop.~\ref{thm:indivisibleCPTPmaps}). We give strong evidence in this paper that it is the divisibility of the process that enforces Tsirelson's bound for our model.

\begin{defn}
A process is $N$-\emph{divisible} with respect to a set of times $\{t_1,t_2,\cdots, t_N\}$ when the maps relating the system state between any two time-steps can be described by a composition of completely positive trace-preserving (CPTP) maps between intermediate times: $\Lambda_{l;j} = \Lambda_{l;k} \circ \Lambda_{k;j}\; \forall \,j,\, k,\, l$ where $t_l> t_k> t_j$.\label{defn:divisible}
\end{defn}
Note that by defining $\Lambda_A\equiv\Lambda_{t_2;t_1}$, $\Lambda_B\equiv\Lambda_{t_4;t_3}$, and $\Lambda_E\equiv\Lambda_{t_3;t_2}$ independently, we are automatically imposing divisible dynamics with respect to $\{t_1,t_2,t_3,t_4\}$. Conversely, the process is indivisible when either $\Lambda_E$ or $\Lambda_B$ depend on Alice's measurement choice, as will be shown below. We will first study the classical bound of Eq.~\eqref{bellfunction} to reveal that it is satisfied if $\Lambda_E$ is an entanglement-breaking channel. We then give strong evidence that divisible dynamics leads to the temporal Tsirelson's bound. Finally, we study indivisible dynamics and its consequences on the temporal Bell's inequality.

\section{Entanglement-breaking dynamics and classicality}

We begin by choosing the channel $\Lambda_E$ to be any entanglement-breaking channel. Given an arbitrary entangled state of a composite system, a channel is entanglement-breaking and trace-preserving (EBT) if and only if its action on a subsystem yields a separable state. We will give strong evidence that, in this case and for $\vec v = 0$ (i.e., the initial state completely mixed), $\mathcal{B}_{EBT} \le \mathcal{B}_{\mathrm{cl}} = 2$, so that we retrieve the well-known classical bound. Entanglement breaking channels include as a subset all stochastic processes.

EBT channels can be viewed as LOCC (local operations and classical communication) channels. In their explicit form, EBT channels first involve a measurement giving some outcome $k$ and then a re-preparation of some state using this outcome $k$. The step between the measurement and re-preparation is classical. With the insertion of such a classical component---in between Alice and Bob, say---there is no entanglement between Alice and Bob. We would expect that the classical bound is obeyed. But this is not always the case.

We first show, by construction, that the assumption of $\vec v = 0$ is necessary, as relaxing it allows the Bell function in Eq.~\eqref{bellfunction} to exceed the classical bound. Take $\Lambda_A$ to be an identity channel, $\Lambda_E$ to be a projective measurement along vector $\vec{c}$ (clearly an EBT channel), and $\Lambda_B$ to always output state $\vec b$ independently of the input state.
Further assume that Bob's measurements are $\vec{b}_{1}=\vec{c}$ and $\vec{b}_{2}=\vec{b}$.
Then one easily verifies that the temporal correlations read: $E_{i3} = \vec a_i \cdot \vec c$ and $E_{i4} = \vec a_i \cdot \vec v$.
With this at hand, the Bell function in Eq.~\eqref{bellfunction}, which we label here $\mathcal{B}_{\mathrm{EBT}}$, reaches its maximum if $\vec{a}_{1}+\vec{a}_{2}$ is parallel to $\vec{c}$ and $\vec{a}_{1}-\vec{a}_{2}$ is parallel to $\vec{v}$, thus
leading to $\mathcal{B}_{EBT}  = 2\sqrt{1+|\vec{v}|^{2}}$. 
The classical bound is thus violated for all input states with $|\vec{v}|>0$. If the state is pure, $|\vec{v}|=1$ then the Tsirelson bound is obtained.

It is well known that if all the channels are identity channels, then the Tsirelson bound can be achieved regardless of the purity of the input system. Identity channels, or equivalently unitary channels, preserve the quantum coherence (and entanglement) of systems they act upon. Entanglement-breaking channels intuitively destroy entanglement and coherence. As the example above shows, the bias or \emph{purity} of the input state acts as a resource: it enables us to violate the classical bound. We will now give strong evidence that in the \emph{absence} of purity, i.e., if  $|\vec{v}|=0$, then EBT channels will enforce the classical bound. First we introduce the notion of unital channels: $\Lambda (\mathbbm{1}) =\mathbbm{1}.$

\begin{prop}
If $\Lambda_A$ is unital and $\Lambda_B$ is an arbitrary CPTP channel, $\Lambda_E$ is EBT, and $\vec v = 0$, then $\mathcal{B}_{EBT}\leq 2$.
\label{thm:EBTchannels}
\end{prop}
\begin{proof}
For the moment we let $\Lambda_A$ be an arbitrary CPTP channel. Any CPTP map acting on a two-level system can be parameterized as $\Lambda(\vec r) = \vec L + \boldsymbol{\lambda} \vec r$, where $\vec L$ is a three-dimensional vector and $\lambda$ is a matrix. We take $\vec L=\vec A$ and $\boldsymbol{\lambda}= \boldsymbol{\alpha}$ ($\vec L=\vec B$ and $\boldsymbol{\lambda}=\boldsymbol{\beta}$) for $\Lambda=\Lambda_A$ ($\Lambda=\Lambda_B$). In the following, we will make use of the following properties of EBT channels: {\it i)} EBT channels form a convex set; {\it ii)} for two-level systems, the extremal points of EBT channels are extremal classical-quantum (extremal CQ) maps (cf. Theorem 5D in Ref.~\cite{Horodecki2003}).
Such extremal channels are one-dimensional projective measurements, and they output pure states. By Theorem 6 of Ref.~\cite{Horodecki2003}, an extremal EBT channel acting on a qubit is fully specified using only two Kraus operators.
Hence, if $\Lambda_E$ is an extremal CQ channel, it can equivalently be thought of as a projective measurement along $\vec{c}$ and a preparation of states with Bloch vectors $\vec{r}_+$ and $\vec{r}_-$. The resulting correlation functions are listed in the~\ref{appendix:CPTPEBTCPTP} and lead to the Bell function
\begin{align}
\mathcal{B}_{EBT} =& \frac{1}{2} \left[ (\vec c \cdot \vec a_2) ( \vec b_1 \cdot \vec s) - ( \vec c \cdot \vec a_2 ) ( \vec b_2 \cdot \boldsymbol{\beta} \vec s ) 
\right. + (\vec{a}_{2} \cdot \vec{A}) ( \vec{b}_{1}\cdot\vec{t}-2\vec{b}_{2}\cdot\vec{B}-\vec{b}_{2}\cdot\boldsymbol{\beta}\vec{t}) 
  \nonumber\\
 & \quad + \left. (\vec{c}\cdot\boldsymbol{\alpha}\vec{a}_{1}) (\vec{b}_{1}\cdot\vec{s}) 
 +(\vec{c}\cdot\boldsymbol{\alpha}\vec{a}_{1}) ( \vec{b}_{2}\cdot\boldsymbol{\beta}\vec{s}) \right],\label{eq:bellfunctionEBT}
\end{align}
where $\vec{s}=\vec{r}_{+}-\vec{r}_{-}$ and $\vec{t} = \vec{r}_{+}+\vec{r}_{-}$. 
{Let us now impose that $\Lambda_A$ is unital (hence $\vec{A}=0$). Next note that $\vec{b}_{2}\cdot (\boldsymbol{\beta}\vec{s}) = (\boldsymbol{\beta}^T \vec{b}_2) \cdot \vec s$, which gives us
\begin{align}
|\mathcal{B}_{EBT}| & = \dfrac{1}{2} \left| \left(\vec{c}\cdot \boldsymbol{\alpha}\vec{a}_{1} \right) \vec{s} \cdot\left(\vec{b}_{1}  + \boldsymbol{\beta}^T \vec{b}_2 \right) +\left(\vec{c}\cdot\vec{a}_{2}\right) \vec{s} \cdot \left( \vec{b}_{1} - \boldsymbol{\beta}^T \vec{b}_2 \right)  \right| \\
& \le \dfrac{1}{2} \left| \vec{c}\cdot \boldsymbol{\alpha}\vec{a}_{1} \right| \left|\vec{s} \cdot\left(\vec{b}_{1} + \boldsymbol{\beta}^T \vec{b}_2 \right) \right| 	+ \frac{1}{2}\left|\vec{c} \cdot\vec{a}_{2}\right| \left|\vec{s} \cdot \left( \vec{b}_{1} - \boldsymbol{\beta}^T \vec{b}_2 \right)  \right| \\
& \le \dfrac{1}{2} \left|\vec{s} \cdot\left(\vec{b}_{1}  +  \boldsymbol{\beta}^T \vec{b}_2 \right) \right| 	+ \frac{1}{2} \left|\vec{s} \cdot \left( \vec{b}_{1} - \boldsymbol{\beta}^T \vec{b}_2 \right)  \right| \\  
& \le \max\left\{\left|\vec{s} \cdot \vec{b}_1 \right|, \left|\vec s \cdot \boldsymbol{\beta}^T\vec{b}_2 \right| \right\} \le 2.
\end{align}
Going from the second line to the third line we used the fact that $|\vec{c}\cdot \boldsymbol{\alpha}\vec{a}_{1}| \le 1$ and $ |\vec{c}\cdot \vec{a}_{2}| \le 1$. In the final line we used the fact $|\vec s| \le 2$, $|\vec{b}_1| \le 1$, and $|\boldsymbol{\beta}^T \vec{b}_2| \le 1$. By convexity, the proof holds when $\Lambda_E$ is an arbitrary EBT channel.}
\end{proof}

\emph{Conjecture. An analytical proof for the upper bound when $\vec A \ne 0$ is non-trivial. However, a constrained numerical optimization shows that $\max|\mathcal{B}_{EBT}| = 2$ even in this case. We therefore conjecture that $\mathcal{B}_{EBT}$ lies within the interval $[-2,2]$. For this optimization all vectors were parameterized in spherical coordinates and we introduced matrix elements $\beta_{ij}$ such that $|\vec B + \beta \vec r| \le 1$ holds for $256$ vectors $\vec r$ distributed uniformly over the sphere, and similarly for $\vec A$ and $\alpha$. As the Bell function in Eq.~\eqref{eq:bellfunctionEBT} is smooth, the numerical optimisation should be robust.}

We have recovered the classical bound for the Bell function using a quantum system initially prepared in a maximally mixed state, where an EBT channel is placed between Alice's and Bob's measurements. The presence of such a channel between Alice and Bob in the evolution of the system is crucial: its absence would in general allow for the attainment of the quantum bound.

Suppose $\Lambda_A$ is an extremal CQ channel with projective measurements along $\vec{c}$ and output states with Bloch vectors $\vec{r}_+$ and $\vec{r}_-$. 
If the input state is maximally mixed (i.e., $\vec v = 0$), the Bell function reads
\begin{gather}
\mathcal{B}_{EBT} = \frac12( \vec{c} \cdot \vec{a}_{1}) (\vec{b}_{1}+\vec{b}_{2}) \cdot (\vec{r}_{+}-\vec{r}_{-})+\vec{a}_{2}\cdot (\vec{b}_{1}-\vec{b}_{2}).
\end{gather}
By choosing $\vec{b}_{1}+\vec{b}_{2} = 2 \cos \theta \, \vec w$ and $\vec{b}_{1}-\vec{b}_{2} = 2 \sin \theta \, \vec w_{\perp}$ with  $\vec w$ and $\vec w_{\perp}$ two orthonormal vectors, one directly verifies that $\mathcal{B}_{EBT}$ can achieve the temporal Tsirelson's bound.

On the other hand, we can still enforce  $\mathcal{B} \le 2$ even when the channel between Alice and Bob is entanglement-preserving. Let $\Lambda_A$ and $\Lambda_B$ both be identity channels, and consider the qubit channel of Werner type $\Phi_W=p \mathcal{I}+\left(1-p\right) \frac{1}{2} \mathbbm{1}, \, p\in\left[0,1\right]$, where $\mathcal{I}$ is the identity channel and $\frac{1}{2}\mathbbm{1}$ is the maximally incoherent channel, that replaces any input state with the maximally mixed one. The channel $\Phi_W$ is EBT if $p<1/3$. Taking $\Lambda_E= \Phi_W$, the Bell function is $\mathcal{B}_{W} =p \mathcal{B}_1$, where $\mathcal{B}_1 \le 2 \sqrt{2}$ is the Bell function corresponding to the identity channel, i.e., for $p=1$. For $1/3<p<1 /\sqrt{2}$ the channel $\Phi_W$ is entanglement-preserving, yet $\mathcal{B}_W < 2$.

We have now given strong evidence for the conditions which guarantee the classical bound $\mathcal{B}_{cl}=2$, namely that the channel connecting Alice and Bob is entanglement breaking and the initial state is completely mixed. We have also shown that the `converse' statement does not hold, i.e., having an entanglement-preserving channel between Alice and Bob and $|\vec{v}|>0$ does not guarantee a violation of the classical bound. We now move to considerations of general quantum channels.

\section{Temporal Tsirelson's bound for divisible processes}

We give strong evidence that for generic CPTP maps, the temporal Bell function still obeys the quantum upper bound $\mathcal{B}_q = 2 \sqrt{2}$. This is done through the following.

\begin{thm}
\label{thm:CPTPdynamics}
For $\Lambda_A$, $\Lambda_E$ arbitrary CPTP channels, and $\Lambda_B$ unital, we have $\mathcal{B}\leq 2\sqrt{2}$. \label{thm:CPTP-Part2}
\end{thm}
\begin{proof}
First let all three channels be arbitrary. We use the same sort of parametrisation for $\Lambda_{A,B,E}$ introduced in the proof of Proposition~\ref{thm:EBTchannels}. The input state is specified by the Bloch vector $\vec v$. We introduce new vectors $\vec \xi_1 = (\vec v \cdot \vec a_1) \vec E + \gamma [(\vec v \cdot \vec a_1) \vec A + \alpha \vec a_1]$ and $\vec \xi_2 =\gamma \vec a_2+ [(\vec A + \alpha \vec v)\cdot \vec a_2] \vec E $. In~\ref{appendix:CPTPdynamics} it is shown that the Bell function can then be rewritten as
\begin{gather}
\mathcal{B}_q  =  \vec b_1 \cdot (\vec \xi_1 + \vec \xi_2) + b_2 \cdot \beta (\vec \xi_1  -  \vec \xi_2) + [\vec v \cdot \vec a_1 - (\vec A + \alpha \vec v)\cdot \vec a_2] (\vec B \cdot \vec b_2).
\label{BQ}
\end{gather}
Note that for any channel $\Lambda(\vec r) = \vec L + \lambda \vec r$, we have $|f \vec L + \lambda \vec r| \le 1$, with $f$ a scalar such that $|f| \le 1$. Therefore, both our new vectors $\vec \xi_{1,2}$ have at most unit length. If we now impose that channel $\Lambda_B$ is unital, i.e., $\vec B = 0$, the last term of Eq.~\eqref{BQ} vanishes and $\mathcal{B}_q$ takes the form of the usual Bell function. Hence $\mathcal{B}_q \le 2 \sqrt{2}$.
\end{proof}

\emph{Conjecture.} We conjecture the last theorem also holds when $\vec B \ne 0$. We have verified numerically that the same bound holds for a non-zero $\vec B$. We also show that if $\Lambda_A$, $\Lambda_B$ are unitary and $\Lambda_E$ is an arbitrary CPTP map, the temporal Tsirelson's bound is achieved only when $\Lambda_E$ is a unitary map (see~\ref{appendix:CPTPdynamics}).

The proofs of both of our main Theorems are partially numerical. We point out that fully analytical arguments would be cumbersome. They are possible in the spatial scenario owing to the convexity of the set of states, which reduces the problem to a proof using pure states only. Furthermore, the use of the Schmidt decomposition reduces the number of variables over which one has to optimise. In the temporal case, on the other hand, a reduction in the number of variables is less forthcoming. Even if we consider only extremal maps, parameter counting shows that there are $30$ variables to be optimised [c.f.~\ref{appendix:CPTPdynamics}]. As such extremal CPTP maps are not unitary, the problem is clearly non-trivial.

We have given strong evidence that any divisible quantum process has its corresponding Bell function bounded from above by Tsirelson's bound. For unitary dynamics this was effectively shown in Refs.~\cite{Taylor2004} and~\cite{Fritz2010}, but now it is clear that the relevant property of unitary transformation is their divisibility. Furthermore, non-unitary dynamics can also lead to correlations that achieve Tsirelson's bound as long as the input state is biased, i.e., $|\vec v| > 0$.

Finally, note that it is under the assumption of divisibility of $\Lambda_A$, $\Lambda_E$, and $\Lambda_B$ we reach Eq.~\eqref{BQ}. If the process is indivisible then all three channels could be a function of Alice and Bob's measurement choice as well as the corresponding outcome. The analog of Eq.~\eqref{BQ} for such a case could have many more parameters. 

\begin{figure}
\centering \includegraphics[width=0.6\columnwidth]{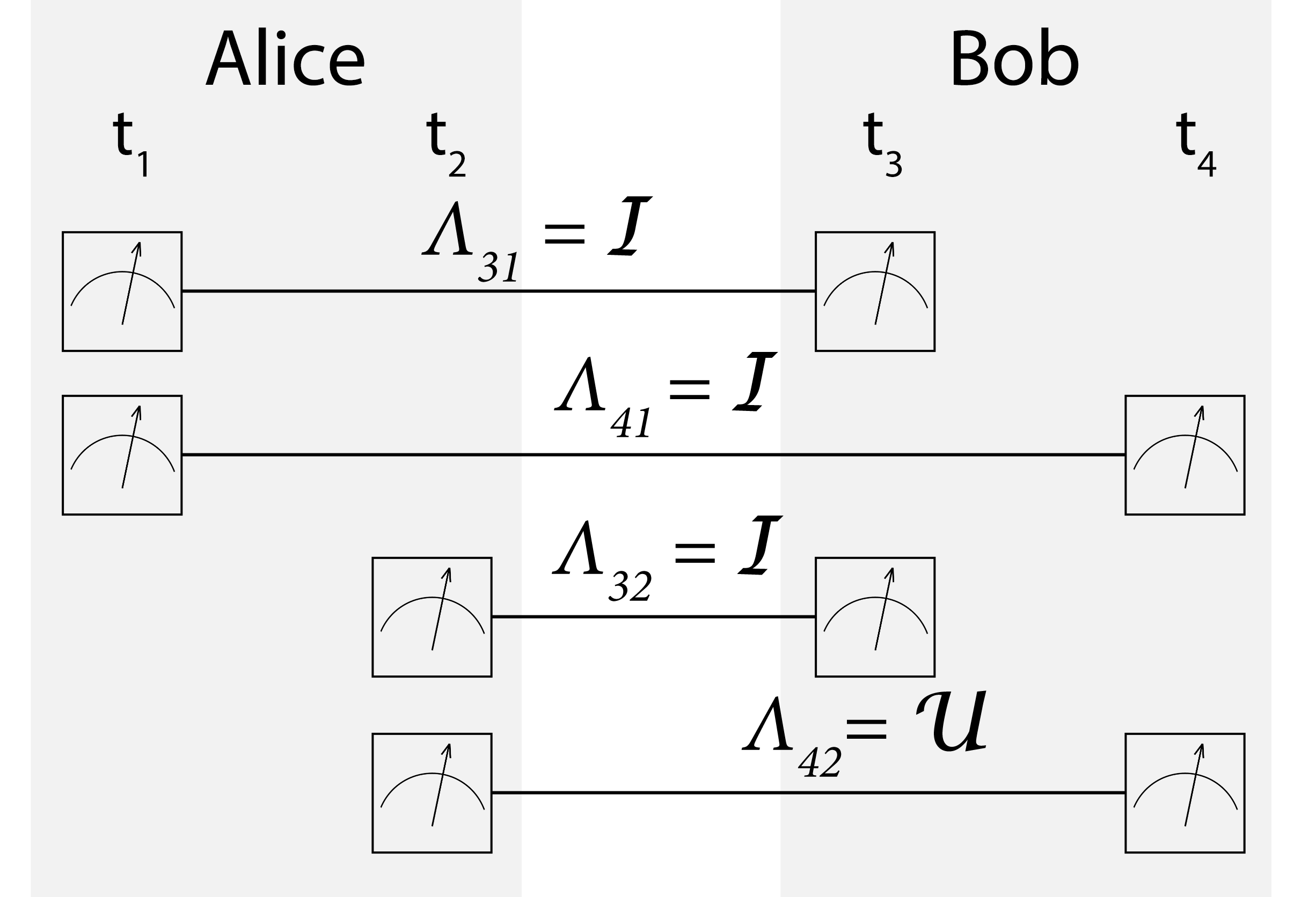}
\protect\caption{Example of an indivisible process. Alice measures either at $t_{1}$ or $t_{_{2}}$, and Bob measures either at $t_{3}$ or $t_{4}$. If Alice measures at $t_{2}$, then there is a unitary channel in between $\left\{ t_{3},t_{4}\right\}$. All other channels are identity channels $\mathcal{I}$. See main text for a proof that this dynamics is not divisible and leads to the maximal algebraic value of the temporal Bell function.
\label{fig:Nondivisiblemap_quantum}}
\end{figure}

\section{Indivisible processes}

A process is indivisible when it does not satisfy Definition~\ref{defn:divisible}. Such processes can be characterised by CPTP transformations using the superchannel formalism~\cite{modi2012operational, PhysRevLett.114.090402}.

\begin{prop}
Some indivisible processes may yield $\mathcal{B}_{ID} > 2\sqrt{2}$. \label{thm:indivisibleCPTPmaps}
\end{prop}
\begin{proof}
In Fig.~\ref{fig:Nondivisiblemap_quantum}, the channels $\Lambda_{31}$, $\Lambda_{41}$ and $\Lambda_{32}$ are identity channels and the last channel $\Lambda_{42}=\mathcal{U}$ is unitary. As the action of the unitary channel on a state is to rotate its Bloch vector, let $R_u$ be the equivalent rotation of $\mathcal{U}$. The Bell function is thus
\begin{gather}
\mathcal{B}_{ID}=(\vec{a}_1+\vec{a}_2)\cdot\vec{b}_1 +[\vec{a}_1-\left(\boldsymbol{R_u}\vec{a}_2\right)]\cdot\vec{b}_2.
\end{gather}
By letting $\vec{a}_1=\vec{a}_2=\vec{b}_1=\vec{b}_2$ and the unitary transformation $\mathcal{U}$ be such that $\boldsymbol{R_u} \vec{a}_2=-\vec{b}_2 =-\vec{a}_2$, we get $\mathcal{B}_{ID} = 4$, i.e., we violate Tsirelson's bound and achieve the maximum algebraic value.

We now prove, by contradiction, that this process is indivisible. As $\Lambda_{32}=\mathcal{I}$ and $\Lambda_{42}=\mathcal{U}$, divisibility would imply that $\Lambda_{43}=\mathcal{U}$. Divisibility of $\Lambda_{43}=\mathcal{U}$ and  $\Lambda_{41}=\mathcal{I}$ then imply that $\Lambda_{31}=\mathcal{U}^\dagger$. This contradicts the original assumption $\Lambda_{31}=\mathcal{I}$. The process is thus indivisible.
\end{proof}

In addition to the above, note that such an indivisible process can also be realised on a classical system. As all the Bloch vectors are either parallel or antiparallel, they encode classical information only. On the other hand, indivisibility is not sufficient to exceed the classical bound $\mathcal{B}_{cl} =2 $. If all the maps are the maximally incoherent map $\mathbbm{1}$, which always outputs $\mathbbm{1}/2$, then $B_{\mathbbm{1}}=0$. Denote the maps of the scenario in Proposition~\ref{thm:indivisibleCPTPmaps} as $\{\mathcal{I}, \mathcal{U}\}$. If we take the convex combination $p\{\mathcal{I}, \mathcal{U}\} + (1-p)\mathbbm{1}$ then $B = pB_{\{\mathcal{I}, \mathcal{U}\}}+ (1-p) B_{\mathbbm{1}} =  pB_{\{\mathcal{I}, \mathcal{U}\}} \leq 4p <2 $ for suitable choices of $p$, yet this remains indivisible.

\begin{table*}[b]
\centering
\begin{tabular}{l||c|c|c}
   & \, No superposition \, & \, Superposition \, & \, Superposition \\ 
   &  &  without input bias & with input bias \\ \hline \hline 
 Classical (EBT) divisible &  2 & 2 &  $2\sqrt{2}$\\ \hline
 Quantum divisible \,\, &  contained in classical & $2\sqrt{2}$ & $2\sqrt{2}$\\ \hline
 Indivisible & 4 & 4 & 4 \\
\end{tabular}
\caption{Summary of the results. The type of channel between Alice and Bob is presented in rows. Divisible classical processes are composed of stochastic maps in a fixed basis, whereas divisible quantum processes are described in Def.~\ref{defn:divisible}. If these maps act on states that are diagonal in a fixed basis (No superposition column) the correlations they allow for of course satisfy the classical bound. We give strong evidence that the same classical bound is satisfied for entanglement-breaking channels (EBT) between Alice and Bob when there is no input bias; If there is an input bias then Tsirelson's bound can be obtained. Composition of quantum maps gives quantum divisible processes, which include unitary dynamics. These processes can at most lead to Tsirelson's bound for the temporal Bell function regardless of input bias. Finally, indivisible processes can achieve the algebraic maximum of the Bell function.}
\label{tab}
\end{table*}

\section{Conclusions}

Contrary to the spatial Bell scenario, it is no longer possible to derive a non-trivial bound on the temporal Bell's inequalities which would be independent of the physical systems themselves. This universality of the original Tsirelson's bound is a consequence of essentially static spatial setting, i.e., the particles are only prepared and measured. Nevertheless we have given strong evidence here for a simple condition on the evolving physical system which guarantees that the temporal Tsirelson's bound is satisfied, our results are concisely outlines in Table 1. For the bound to hold the dynamics has to be divisible. Channel divisibility therefore plays a role for correlations in time similar to that played by information causality, macroscopic locality, etc. in space-like scenarios~\cite{infocaus, ProcRolSocA.466.881, NJP.14.063024, FoundPhys.43.805, D.14.239, PRL.112.040401, arXiv1507.07514}. As a consequence, using an argument similar to that in Ref.~\cite{infocaus}, intrinsically indivisible dynamics could be used to violate communication complexity in time. 

Another interesting consideration is that a unitary process looks like an indivisible one from the classical perspective~\cite{PhysRevA.71.032325}. For instance, consider a three time-step process, where a measurement over the basis of eigenstates of $\sigma_z$ is made at any two time steps and between each time step the Hadamard gate is applied. The correlation functions involving the middle time step always vanish, while the correlation function between the initial and final time steps is $1$. At the level of measurement outcomes, the dynamics involving the middle time step are described by fully noisy maps, while the evolution between the initial and final time steps is described by the identity channel. Therefore, from the classical perspective the channel is indivisible, while from the quantum perspective it is perfectly divisible. One can thus conjecture that such a distinction is responsible for the violation of the classical bound of temporal Bell's inequalities and for reaching Tsirelson's bound with unitary processes. In Ref.~\cite{PhysRevLett.113.050401} two-level Leggett-Garg inequalities are constructed from unitary dynamics on multilevel systems. In this case too, the dynamics of the `two levels' will not be divisible as the extra levels of the system act like a structured environment that carry memory.

Finally, note that the indivisible process in Fig.~\ref{fig:Nondivisiblemap_quantum} is no-signalling if the input state is maximally mixed, i.e., the outcomes of Bob do not reveal any information about the settings of Alice. Hence, it is not the `signalling' that maximises the Bell function, but rather it is the non-Markovian memory of the process~\cite{arXiv:1512.00589}. In fact, we never imposed a no-signalling condition even for divisible processes; unlike in space-like correlated systems, time-like processes can of course carry information forward (from Alice to Bob).

\section{Acknowledgments}

This work is supported by the National Research Foundation and Ministry of Education of Singapore Grant No. RG98/13. MP is supported by the EU FP7 grant TherMiQ (Grant Agreement 618074), the John Templeton Foundation (Grant No. 43467), and the UK EPSRC (EP/M003019/1).

\bibliography{paper}

\appendix

\section[Explicit correlation functions of Proposition 2]{Explicit correlation functions when $\Lambda_E$ is EBT (Part of Proposition 2)}\label{appendix:CPTPEBTCPTP}

Let $\Lambda_A$ and $\Lambda_B$ be CPTP and $\Lambda_E$ be EBT. We parametrise them as follows:
\begin{align}
\Lambda_{A}\left(\dfrac{1}{2}\left(\mathbbm{1}+\vec{\sigma}\cdot\vec{r}\right)\right) & =\dfrac{1}{2}\left(\mathbbm{1}+\vec{\sigma}\cdot\left(\vec{A}+\boldsymbol{\alpha}\vec{r}\right)\right), \\
\Lambda_{B}\left(\dfrac{1}{2}\left(\mathbbm{1}+\vec{\sigma}\cdot\vec{r}\right)\right) & =\dfrac{1}{2}\left(\mathbbm{1}+\vec{\sigma}\cdot\left(\vec{B}+\boldsymbol{\beta}\vec{r}\right)\right),\\
\Lambda_{E}\left(\rho\right) & =\sum_{m=\pm 1} R_m \tr \left[{P}_{\vec{c}}^m \rho \right],
\end{align}
where $\vec{A}$, and $\vec{B}$ are vectors, and $\boldsymbol{\alpha}$ and $\boldsymbol{\beta}$ are matrices such that $| \vec{A}+\boldsymbol{\alpha}\vec{v}| \leq1$ etc for any $| \vec{v}| \leq1$. The pure states $R_m$ have Bloch vectors $\vec{r}_\pm$ for $m = \pm 1$. Furthermore, denote by $P_{\vec{a}_{i}}^{k} = \frac{1}{2}(\mathbbm{1} + k \vec{\sigma} \cdot \vec a_i)$ the post-measurement state of Alice if she obtains result $k = \pm$ at time $t_i$. Similarly, $P_{\vec{b}_{j}}^{l} = \frac{1}{2}(\mathbbm{1} + l \vec{\sigma} \cdot \vec b_j)$ denotes the post-measurement state of Bob if he observes outcome $l = \pm 1$ at time $t_j$. With this notation the correlation functions read:
\begin{align}
E_{13} =& \sum_{k,\ell=\pm1}k\cdot\ell\cdot\tr\left[\rho P_{\vec{a}_{1}}^{k} \right] \tr\left[\Lambda_{EBT}\left(\Lambda_{A}\left(P_{\vec{a}_{1}}^{k}\right)\right)P_{\vec{b}_{1}}^{\ell}\right]\nonumber\\
 =&\dfrac{1}{2}\left\{ \left(\vec{c}\cdot\boldsymbol{\alpha}\vec{a}_{1}\right)\vec{b}_{1}\cdot\left(\vec{r}_{+}-\vec{r}_{-}\right)+\left(\vec{v}\cdot\vec{a}_{1}\right)\vec{b}_{1}\cdot\left(\vec{r}_{+}+\vec{r}_{-}\right)\right\}\nonumber\\
&+\dfrac{1}{2}\left\{\left(\vec{v}\cdot\vec{a}_{1}\right)\left(\vec{c}\cdot\vec{A}\right)\vec{b}_{1}\cdot\left(\vec{r}_{+}-\vec{r}_{-}\right)\right\}
\end{align}
\begin{align}
E_{14} =&\sum_{k,\ell=\pm1}k\cdot\ell\cdot\tr\left[\rho P_{\vec{a}_{1}}^{k}\right]\tr\left[\Lambda_{B}\left(\Lambda_{EBT}\left(\Lambda_{A}\left(P_{\vec{a}_{1}}^{k}\right)\right)\right)P_{\vec{b}_{2}}^{\ell}\right]\nonumber\\
 =&\dfrac{1}{2}\left(\vec{c}\cdot\boldsymbol{\alpha}\vec{a}_{1}\right)\vec{b}_{2}\cdot\boldsymbol{\beta}\left(\vec{r}_{+}-\vec{r}_{-}\right)+\left(\vec{v}\cdot\vec{a}_{1}\right)\left(\vec{b}_{2}\cdot\vec{B}\right)\nonumber\\
&+\dfrac{1}{2}\left(\vec{v}\cdot\vec{a}_{1}\right)\vec{b}_{2}\cdot\boldsymbol{\beta}\left(\vec{r}_{+}\left(1+\vec{c}\cdot\vec{A}\right)+\vec{r}_{-}\left(1-\vec{c}\cdot\vec{A}\right)\right)
\end{align}
\begin{align}
E_{23} =&\sum_{k,\ell=\pm1}k\cdot\ell\cdot\tr\left[\Lambda_{A}\left(\rho\right)P_{\vec{a}_{2}}^{k}\right]\tr\left[\Lambda_{EBT}\left(P_{\vec{a}_{2}}^{k}\right)P_{\vec{b}_{1}}^{\ell}\right]\nonumber\\
 =&\dfrac{1}{2}\left\{ \left(\vec{c}\cdot\vec{a}_{2}\right)\vec{b}_{1}\cdot\left(\vec{r}_{+}-\vec{r}_{-}\right)+\left(\vec{a}_{2}\cdot\vec{w}\right)\vec{b}_{1}\cdot\left(\vec{r}_{+}+\vec{r}_{-}\right)\right\}
\end{align}
\begin{align}
E_{24} =&\sum_{k,\ell=\pm1}k\cdot\ell\cdot\tr\left[\Lambda_{A}\left(\rho\right)P_{\vec{a}_{2}}^{k}\right]\tr\left[\Lambda_{B}\left(\Lambda_{EBT}\left(P_{\vec{a}_{2}}^{k}\right)\right)P_{\vec{b}_{2}}^{\ell}\right]\nonumber\\
 =& \dfrac{1}{2}\left(\vec{c}\cdot\vec{a}_{2}\right)\vec{b}_{2}\cdot\boldsymbol{\beta}\left(\vec{r}_{+}-\vec{r}_{-}\right)+\left(\vec{a}_{2}\cdot\vec{w}\right)\left(\vec{b}_{2}\cdot\vec{B}\right)\nonumber\\
 &+\dfrac{1}{2}\left(\vec{a}_{2}\cdot\vec{w}\right)\vec{b}_{2}\cdot\boldsymbol{\beta}\left(\vec{r}_{+}+\vec{r}_{-}\right),
\end{align}
where $\vec{w}=\vec{A}+\boldsymbol{\alpha}\vec{v}$.


\section{Derivation of Eq. (4) in Theorem 3 }\label{appendix:CPTPdynamics}

Let us parameterise arbitrary CPTP maps, $\Lambda_{A}$, $\Lambda_{E}$ and $\Lambda_{B}$ as
\begin{align}
\Lambda_{A}\left(\dfrac{1}{2}\left(\mathbbm{1}+\vec{\sigma}\cdot\vec{r}\right)\right) & =\dfrac{1}{2}\left(\mathbbm{1}+\vec{\sigma}\cdot\left(\vec{A}+\boldsymbol{\alpha}\vec{r}\right)\right)\\
\Lambda_{E}\left(\dfrac{1}{2}\left(\mathbbm{1}+\vec{\sigma}\cdot\vec{r}\right)\right) & =\dfrac{1}{2}\left(\mathbbm{1}+\vec{\sigma}\cdot\left(\vec{E}+\boldsymbol{\gamma}\vec{r}\right)\right)\\
\Lambda_{B}\left(\dfrac{1}{2}\left(\mathbbm{1}+\vec{\sigma}\cdot\vec{r}\right)\right) & =\dfrac{1}{2}\left(\mathbbm{1}+\vec{\sigma}\cdot\left(\vec{B}+\boldsymbol{\beta}\vec{r}\right)\right)
\end{align}
where $\vec{A}$, $\vec{E}$ and $\vec{B}$ are vectors, and $\boldsymbol{\alpha}$, $\boldsymbol{\gamma}$ and $\boldsymbol{\beta}$ are matrices such that $| \vec{A}+\boldsymbol{\alpha}\vec{r}| \leq1$
etc for any $| \vec{r}| \leq1$.
The correlation functions are
\begin{align}
E_{13} & =\sum_{k,\ell=\pm1}k\cdot\ell\tr\left[\rho P_{\vec{a}_{1}}^{k}\right]\tr\left[\Lambda_{E}\left(\Lambda_{A}\left(P_{\vec{a}_{1}}^{k}\right)\right)P_{\vec{b}_{1}}^{\ell}\right] \nonumber\\
 & =\left(\boldsymbol{\gamma\alpha}\vec{a}_{1}\right)\cdot\vec{b}_{1}+\left(\vec{v}\cdot\vec{a}_{1}\right)\left[\vec{E}+\boldsymbol{\gamma}\vec{A}\right]\cdot\vec{b}_{1}\\
E_{14} & =\sum_{k,\ell=\pm1}k\cdot\ell\tr\left[\rho P_{\vec{a}_{1}}^{k}\right]\tr\left[\Lambda_{B}\left(\Lambda_{E}\left(\Lambda_{A}\left(P_{\vec{a}_{1}}^{k}\right)\right)\right)P_{\vec{b}_{2}}^{\ell}\right] \nonumber\\
 & =\left(\boldsymbol{\beta\gamma\alpha}\vec{a}_{1}\right)\cdot\vec{b}_{2}+\left(\vec{v}\cdot\vec{a}_{1}\right)\left[\vec{B}+\boldsymbol{\beta}\left(\vec{E}+\boldsymbol{\gamma}\vec{A}\right)\right]\cdot\vec{b}_{2}\\
E_{23} & =\sum_{k,\ell=\pm1}k\cdot\ell\tr\left[\Lambda_{A}\left(\rho\right)P_{\vec{a}_{2}}^{k}\right]\tr\left[\Lambda_{E}\left(P_{\vec{a}_{2}}^{k}\right)P_{\vec{b}_{1}}^{\ell}\right] \nonumber\\
 & =\left(\boldsymbol{\gamma}\vec{a}_{2}\right)\cdot\vec{b}_{1}+\left[\vec{A}+\boldsymbol{\alpha}\vec{v}\right]\cdot\vec{a}_{2}\left(\vec{E}\cdot\vec{b}_{1}\right)\\
E_{24} & =\sum_{k,\ell=\pm1}k\cdot\ell\tr\left[\Lambda_{A}\left(\rho\right)P_{\vec{a}_{2}}^{k}\right]\tr\left[\Lambda_{B}\left(\Lambda_{E}\left(P_{\vec{a}_{2}}^{k}\right)\right)P_{\vec{b}_{2}}^{\ell}\right] \nonumber\\
 & =\left(\boldsymbol{\beta\gamma}\vec{a}_{2}\right)\cdot\vec{b}_{2}+\left[\vec{A}+\boldsymbol{\alpha}\vec{v}\right]\cdot\vec{a}_{2}\left[\vec{B}+\boldsymbol{\beta}\vec{E}\right]\cdot\vec{b}_{2}.
\end{align}
Hence, the Bell function is 
\begin{align}
\mathcal{B}_{Q} 
=&\left(\boldsymbol{\gamma\alpha}\vec{a}_{1}+\boldsymbol{\gamma}\vec{a}_{2}\right)\cdot\vec{b}_{1}+\boldsymbol{\beta}\left(\boldsymbol{\gamma\alpha}\vec{a}_{1}-\boldsymbol{\gamma}\vec{a}_{2}\right)\cdot\vec{b}_{2}\nonumber\\
 & +\left(\vec{v}\cdot\vec{a}_{1}\right)\left\{ \left[\vec{E}+\boldsymbol{\gamma}\vec{A}\right]\cdot\vec{b}_{1}+\left[\vec{B}+\boldsymbol{\beta}\left(\vec{E}+\boldsymbol{\gamma}\vec{A}\right)\right]\cdot\vec{b}_{2}\right\} \nonumber\\
 & +\left[\vec{A}+\boldsymbol{\alpha}\vec{v}\right]\cdot\vec{a}_{2}\left\{ \vec{E}\cdot\vec{b}_{1}-\left[\vec{B}+\boldsymbol{\beta}\vec{E}\right]\cdot\vec{b}_{2}\right\}.
\end{align}
Substituting the variables $\vec \xi_1$ and $\vec \xi_2$ defined in the main text one directly recovers Eq. (4) of the main text.


\section[Tsirelson's bound is achieved when all maps are unitary]{For unitary $\Lambda_A$ and $\Lambda_B$, Tsirelson's bound is achieved only when $\Lambda_E$ is also unitary}\label{appendix:CPTPandIdentity}

Let $\rho$ be the input density matrix with Bloch vector $\vec{v}$ and let $\Lambda_{E}$ CPTP parameterised by matrix $\boldsymbol{\gamma}$ and shift vector $\vec{E}$. Let $\Lambda_{A}$ and $\Lambda_{B}$ be unitary channels represented by the matrices $\boldsymbol{\alpha}$ and $\boldsymbol{\beta}$, respectively. 
The Bell's parameter reads
\begin{align}
\mathcal{B}_{CPTP} = & \left(\boldsymbol{\gamma}\boldsymbol{\alpha}\vec{a}_{1}+\boldsymbol{\gamma}\vec{a}_{2}\right)\cdot\vec{b}_{1}+\left(\boldsymbol{\beta}\boldsymbol{\gamma}\boldsymbol{\alpha}\vec{a}_{1}+\boldsymbol{\beta}\boldsymbol{\gamma}\vec{a}_{2}\right)\cdot\vec{b}_{2} \nonumber\\
 &  +\left(\vec{v}\cdot\vec{a}_{1}\right)\left(\vec{E}\cdot\vec{b}_{1}+\boldsymbol{\beta}\vec{E}\cdot\vec{b}_{2}\right) \nonumber\\
 & +\left(\boldsymbol{\alpha}\vec{v}\cdot\vec{a}_{2}\right)\left(\vec{E}\cdot\vec{b}_{1}-\boldsymbol{\beta}\vec{E}\cdot\vec{b}_{2}\right)
\end{align}
which simplifies to
\begin{gather}
\mathcal{B}_{CPTP}=\vec \eta_1 \cdot (\vec{b}_{1}+\vec{b}_{2}^{\prime}) + \vec \eta_2 \cdot (\vec{b}_{1}-\vec{b}_{2}^{\prime}).\label{eq:blochsphereBell}
\end{gather}
after introduction of rotated vectors $\vec{a}_{1}^{\prime}=\boldsymbol{\alpha}\vec{a}_{1}$, $\vec{v}^{\prime} = \boldsymbol{\alpha} \vec{v}$ and $\vec{b}_{2}^{\prime}=\boldsymbol{\beta}^{T}\vec{b}_{2}$, and where $\vec \eta_1 = (\vec{v}^{\prime} \cdot \vec{a}_{1}^{\prime}) \vec E + \boldsymbol{\gamma} \vec{a}_{1}^{\prime}$ and $\vec \eta_2 = (\vec{v}^{\prime} \cdot \vec{a}_{2}) \vec{E} +  \boldsymbol{\gamma} \vec{a}_{2}$. Since vectors $\vec{b}_{1}+\vec{b}_{2}^{\prime}$ and $\vec{b}_{1}-\vec{b}_{2}^{\prime}$ are orthogonal and their lengths can be parameterised by a single angle $|\vec{b}_{1}+\vec{b}_{2}^{\prime}| = 2 \sin \theta$ and $|\vec{b}_{1}-\vec{b}_{2}^{\prime}| = 2 \cos \theta$, the Tsirelson's bound can only be achieved if both $\vec \eta_1$ and $\vec \eta_2$ are unit vectors. Since for all channels satisfying $|\vec E + \gamma \vec r | < 1$ for all unit vectors $\vec r$, also $|\vec \eta_i| < 1$, we are left with studies of channels that output at least one pure state.

All completely positive, trace preserving maps can be reduced to the form (up to unitary conjugation before and after the map, which does not affect the value of the Bell's parameter): 
\begin{align}
\boldsymbol{\gamma} & =\left[\begin{array}{ccc}
\cos\theta & 0 & 0\\
0 & \cos\phi & 0\\
0 & 0 & \cos\theta\cos\phi
\end{array}\right],\,
\vec{E}=\left[\begin{array}{c}
0\\
0\\
\sin\theta\sin\phi
\end{array}\right],
\end{align}
with $\theta\in\left[0,2\pi\right)$, $\phi\in\left[0,\pi\right)$. If $\vec{E}=\vec{0}$, then we must have $\sin\theta=0$ (or $\sin\phi = 0$) and correspondingly $\cos\theta=\pm1$ (or $\cos\phi=\pm1$). Thus, there is only one direction along which $\gamma$ does not shrink its input vector. Therefore, $\gamma \vec a_1' = \pm \gamma \vec a_2$ and the Bell's parameter is restricted by its classical bound as Alice effectively chooses only one setting.

If $|\vec{E}| > 0$, then $|\sin\theta| > 0$ and $|\sin\phi| > 0$. Matrix $\gamma$ applied on an arbitrary unit vector $\vec r$ gives now an ellipsoid of radius \emph{strictly less} than $1$. Furthermore, shifting such obtained vectors by $\vec E$ has to result in a new vector with $|\vec E + \gamma \vec r| \le 1$ in order to guarantee that physical states are mapped to physical states. If instead of shifting by $\vec E$ we now shift by $(\vec v \cdot \vec a) \vec E$, as in our Bell's parameter, the resulting vectors will be shorter than a unit vector except when $|\vec v \cdot \vec a_1'| = |\vec v \cdot \vec a_2| = 1$. In this case, however, the settings of Alice are again the same, along $\vec v$, and therefore the temporal Bell's parameter satisfies the classical bound.

Summing up, the maximal violation can only occur if the channel in-between Alice and Bob is unitary.


\section{Parameter counting}
\label{APP_PARAM_COUNTING}

CPTP maps can be written using Kraus operators: $\mathcal{E}\left(\rho\right)=\sum_{i}A_{i}\rho A_{i}^{\dagger}$. Extremal qubit maps have Kraus rank $2$.
In Ref.~\cite{arXiv:quant-ph/0202124v2} the operators could be reduced down to the following form $A_i = U S_i V^\dag$ with $U,V$ unitary and
\begin{gather}
S_{1}=\begin{pmatrix}s & 0\\
0 & t
\end{pmatrix}, \quad
S_{2}=\begin{pmatrix}0 & \sqrt{1-t^{2}}\\
\sqrt{1-s^{2}} & 0
\end{pmatrix},
\end{gather}
where $0\leq s,t\leq1$. 
Each unitary map introduces four parameters, so each extremal map has ten parameters. Three such maps, for each time step, gives $30$ parameters.


\end{document}